\documentclass[12pt]{article}

\usepackage{latexsym}
\usepackage{epsfig,graphics}
\usepackage{amsmath,amssymb,amsthm}
\usepackage{graphics,epsfig,calc}
\usepackage{upref}
 \newcommand{\beqn}{\begin{eqnarray}}
 \newcommand{\eeqn}{\end{eqnarray}}
 \newcommand{\be}{\begin{equation}}
 \newcommand{\ee}{\end{equation}}
 \newcommand{\ba}{\begin{array}}
 \newcommand{\ea}{\end{array}}
\newcommand{\br}{\begin{remark}}
 \newcommand{\er}{\end{remark}}
\newcommand{\bc}{\begin{cor}}
 \newcommand{\ec}{\end{cor}}

 \newcommand{\re}{\ref}
 \newcommand{\ci}{\cite}
 \newcommand{\ds}{\displaystyle}
 \newcommand{\la}{\label}

 \newcommand{\supp}{{\rm supp~}}

\newcommand{\fr}{\frac}

\newcommand{\eps}{\epsilon}

\newcommand{\ti}{\tilde}

\newcommand{\cH}{{\cal H}}

\newcommand{\cO}{{\cal O}}

\newcommand{\ve}{\varepsilon}
\newcommand{\vp}{\varphi}

\newcommand{\de}{\delta}

\newcommand{\al}{\alpha}
\newcommand{\ga}{\gamma}

\newcommand{\si}{\sigma}
\newcommand{\om}{\omega}

\newcommand{\lam}{\lambda}

\newcommand{\ka}{\kappa}

\newcommand{\N}{\mathbb{N}}
\newcommand{\R}{\mathbb{R}}

\newcommand{\nc}{\newcommand}
\nc{\dps}{\displaystyle}

\nc{\RR}{\mbox{\rm I$\!$R}}

\newcommand{\bd}{\begin{defin}}
 \newcommand{\ed}{\end{defin}}
\newcommand{\bt}{\begin{theorem}}
 \newcommand{\et}{\end{theorem}}
\newcommand{\bqt}{\begin{qtheorem}}
 \newcommand{\eqt}{\end{qtheorem}}

\newcommand{\bl}{\begin{lemma}}
 \newcommand{\el}{\end{lemma}}

 \newcommand{\bce}{\begin{center}}
 \newcommand{\ece}{\end{center}}

\newcommand{\bex}{\begin{example}}
 \newcommand{\eex}{\end{example}}
\newcommand{\bexs}{\begin{examples}}
 \newcommand{\eexs}{\end{examples}}

\newcommand{\bexe}{\begin{exercice}}
 \newcommand{\eexe}{\end{exercice}}

\newcommand{\brs}{\begin{remarks}}
 \newcommand{\ers}{\end{remarks}}

\newtheorem{theorem}{Theorem}[section]
\newtheorem{qtheorem}{QTheorem}[section]

\newtheorem{defin}[theorem]{Definition}

\newtheorem{lemma}[theorem]{Lemma}
\newtheorem{remark}[theorem]{Remark}
\newtheorem{remarks}[theorem]{Remarks}
\newtheorem{cor}[theorem]{Corollary}
\newtheorem{pro}[theorem]{Proposition}
\newcommand{\bp}{\begin{pro}}
\newcommand{\ep}{\end{pro}}

\newcommand{\const}{\mathop{\rm const}\nolimits}
\newcommand{\sgn}{\mathop{\rm sgn}\nolimits}

\begin{document}
\begin{titlepage}
\bigskip\bigskip\bigskip

\begin{center}
{\Large\bf
On nonlinear wave equations with
\\~\\
parabolic potentials
}
\vspace{1cm}
 \bigskip\bigskip\\
{\large A.I. Komech}
\footnote{
Supported partly
by Alexander von Humboldt Research Award, and the grants of
DFG, FWF and RFBR.}\\
{\it Fakult\"at f\"ur Mathematik, Universit\"at Wien\\
and Institute for the Information Transmission Problems RAS
}\\
 e-mail:~alexander.komech@univie.ac.at
\medskip\\
{\large E.~A.~Kopylova}
\footnote{Supported partly by the 
Austrian Science Fund (FWF): M1329-N13,
and RFBR grants.}\\
{\it Fakult\"at f\"ur Mathematik, Universit\"at Wien\\
and Institute for Information Transmission Problems RAS}\\
e-mail:~elena.kopylova@univie.ac.at
\medskip\\
{\large S.~A.~Kopylov}
\\
{\it Russian State University of Tourism and Service }\\
e-mail:~covector@yandex.ru 
\end{center}
\date{24 February 2012}

\vspace{1cm}
\begin{abstract}
We introduce a new class of piece-wise quadratic
potentials for  nonlinear wave equations
with a kink solutions.
The potentials allow an exact description
of the spectral properties
for the linearized equation at the kink.
This description is necessary for the
study of the stability
properties of the kinks.

In particular, we construct  examples of the potentials
of Ginzburg-Landau type providing the
asymptotic stability of the kinks \ci{KK1,KK2}.

\end{abstract}



\end{titlepage}


\setcounter{equation}{0}
\section{Introduction}
Last two decades there was an outstanding  activity in the
field of asymptotic stability of solitary waves
for nonlinear Schr\"odinger equations
\ci{BP2, BS2003,Sigal93, SW1,SW2,SW04,TY02,Tsai2003,W85},
nonlinear Klein-Gordon equations
\ci{IKV05,SW99}, relativistic Ginzburg-Landau equations
\ci{KK1,KK2}, and other Hamiltonian PDEs
  \ci{MW96,PW94}. All these results rely on  different
assumptions on the spectral properties
of the corresponding linearized equations.
On the other hand, the  examples were mostly
unknown. Here we construct a model nonlinear
wave equations, providing different spectral properties:
given number of the eigenvalues, absence of the resonances,
and Fermi Golden Rule.

In particular, we construct the examples of
relativistic Ginzburg-Landau equations
providing all properties assumed in
\ci{KK1,KK2}.
\medskip

We consider real solutions to 1D nonlinear Ginzburg-Landau equations
\be\la{Se}
\ddot\psi(x,t)=\psi''(x,t)+F(\psi(x,t)),~~~~~x\in\R
\ee
where  $F(\psi)=-U'(\psi)$.
We assume the following condition.
\medskip\\
{\bf Condition U1} {\it
For some $K>3$ and $m>0$ the potential $U(\psi)$
is smooth even function satisfying}
\beqn\nonumber
&&U(\psi)>0~~~~\mbox{for}~~~\psi\ne a,\\
\la{U11}\\
\nonumber
&&
U(\psi)=\fr{m^2}2 (\psi\mp a)^2+{\cal O}(|\psi\mp a|^{K}),~~
\psi\to\pm a.\la{U12}
\eeqn
The corresponding stationary equation reads
\be\la{ssol}
s''(x)-U'(s(x))=0,~~~~~~x\in\R.
\ee
Constant stationary solutions are: $\psi(x)\equiv 0$
and  $\psi(x)\equiv \pm a$. There are also the ``kinks'', i.e.
nonconstant finite energy solutions $s(x)$
to (\re{ssol})
such that
$$
s(x)\to \pm a,~~~~x\to\pm\infty
$$
Condition {\bf U1} implies that
$(s(x) \mp a)''\sim m^2(s(x) \mp a)$ for
$x\to\pm\infty$,
hence
\be\la{kina}
|s(x) \mp a|\sim Ce^{-m|x|},~~~~~x\to\pm\infty.
\ee
Due to relativistic invariance of
equation (\re{Se}) the  moving kinks
$$
s_{q,v}(x,t)=s(\ka(x-vt-q)),\quad q,v\in\R,\quad  |v|<1,
\quad\ka=1/\sqrt{1-v^2}
$$
also are the solutions to (\re{Se}).
Let us linearize equation (\re{Se}) at the kink $s(x)$.
Substituting $\psi(x,t)=s(x)+\phi(x,t)$,
we obtain formally
$$
  \ddot\phi(x,t)=-H\phi(x,t)+\cO(|\phi(x,t)|^2),
$$
where $H$ is  the Schr\"odinger operator
$$
  H:=-\fr{d^2}{dx^2}+m^2+V(x)
$$
with the potential
$$
  V(x)=-F'(s(x))-m^2=U''(s(x))-m^2.
$$
The condition {\bf U1} and the asymptotics (\re{kina}) imply that
$$
|V(x)|=\cO(|s(x) \mp a|^{K-1})
\sim Ce^{-(K-1)m|x|},~~~~~x\to\pm\infty.
$$
The next properties of $H$ are  valid:
\medskip\\
{\bf H1.} The continuous spectrum of  $H$ is
$\si_cH=[m^2,\infty)$.
\medskip\\
{\bf H2.} The point $\lam_0=0$ belongs to the discrete
spectrum, and corresponding eigenfunction is $s'(x)$.
\medskip\\
{\bf H3.} Since $s'(x)>0$, the point $\lam_0=0$ is
the  groundstate, and all remaining discrete spectrum
is contained in
$(0,m^2]$.
\medskip
\\
To establish an asymptotic stability of the kinks $s_{q,v}(x,t)$
one need  certain  spectral properties of $H$
(cf. \ci{KK1}, \ci{KK2}):
\\\\
{\bf Condition U2} 
{\it The edge point $\lam=m^2$ of the continuous spectrum
is neither eigenvalue nor resonance.
}
\medskip\\
{\bf Condition U3}
{\it The discrete spectrum of $H$ consists of two points: 
$\lam_0=0$ and $\lam_1\in (0,m^2)$ satisfying}
\be\la{U2}
  4\lam_1>m^2.
\ee
We assume also a non-degeneracy condition known as
``Fermi Golden Rule'' meaning  the strong coupling
of the nonlinear term to  the continuous
spectrum. This coupling provides the energy radiation to infinity
(cf. condition (10.0.11) in \ci{BS2003} and condition (1.11)
in \ci{KK2}).
\medskip\\
{\bf Condition U4}
{\it The inequality holds}
\be\la{FGR}
  \int\vp_{4\lam_1}(x)F''(s(x))\vp_{\lam_1}^2(x)dx\ne 0.
\ee
{\it where $\vp_{4\lam_1}$ is the nonzero odd solution to}
$H\vp_{4\lam_1}=4\lam_1\vp_{4\lam_1}$.
\medskip\\
Note that the known quartic double well  Ginzburg-Landau potential  
$U_{GL}(\psi)=(\psi^2-a^2)^2/(4a^2)$ satisfies condition
{\bf U1} with $m^2=2$ and $K=3$ 
as well as conditions {\bf U3}-{\bf U4}.
However, there exist the  resonance 
for the corresponding operator $H$ at the edge point $\lam=m^2$. 
Hence, the  asymptotic stability of the kinks 
for  $U_{GL}$ is the open problem.

Our main result
is the following theorem.

\begin{theorem}\la{t1}
There exist potentials $U(\psi)$
satisfying  conditions {\bf U1}-{\bf U4}.
\end{theorem}
\setcounter{equation}{0}
\section{Piece wise parabolic potentials}
As a first step, we will consider the class of the potentials which are
piece-wise second order polynomials.
\be\la{pp}
U_0(\psi)=\left\{
\ba{ll}
\ds\fr12-\fr b2 \psi^2,& |\psi|\le \ga\\
\\
\ds\fr d2 (\psi\mp 1)^2,& \pm\psi\ge \ga
\ea
\right.
\ee
with some constants $b,d, \ga>0$. 
Let us find the parameters $\ga,b,d$ providing $U_0(\psi)\in C^1(\R)$.
We have
$$
U_0( \ga)=\fr 12 -\fr b2 \ga^2=\fr d2(\ga-1)^2, ~~~~~~
U_0'( \ga)=-b \ga=d( \ga-1).
$$
Solving the equations, we obtain
\be\la{ab}
b=\fr 1 \ga,~~~~~~~d=\fr 1{1- \ga},~~~~~~~~~~~~~0< \ga<1.
\ee
Then the functions
$U_0''(\psi)$ are piece-wise constant
with the jumps at the points $\psi=\pm \ga$.
Thus, the potentials $U_0\in C^1(\R)$ form
one-dimensional manifold parametrized by $\ga\in (0,1)$.
\subsection{Kinks }
Let us solve the equation of type (\re{ssol}) for the kink
in the case of potential (\re{pp}):
\be\la{ssol0}
s_0''(x)-U_0'(s_0(x))=0,~~~~~~x\in\R.
\ee
We search an odd solution to
$$
s_0''(x)=\left\{\ba{ll}
-bs_0(x),&0<s_0(x)\le \ga,
\\\\
d(s_0(x)-1),& s_0(x) > \ga.
\ea\right.
$$
We have
\be\la{kin}
s_0(x)=\left\{\ba{ll}
C\sin\sqrt {b}x,&0<x\le q,
\\\\
Ae^{-\sqrt{d}x}+1,& x > q,
\ea\right.
\ee
where $C> \ga$, $A<0$, $q=\ds\fr{1}{\sqrt b}\arcsin\frac{ \ga}{C}$.
Equating the values of $s(x)$ and its derivative at $x=q$ we obtain
\be\la{ss}
\left\{\ba{rl}
Ae^{-\sqrt{d}q}+1=C\sin\sqrt {b}q=\ga,
\\\\
-\sqrt{d}Ae^{-\sqrt{d}q}=\sqrt {b}C\cos\sqrt {b}q.
\ea\right.
\ee
The first line of (\re{ss}) implies $Ae^{-\sqrt{d}q}= \ga-1$.
Hence the second line  of (\re{ss}) becomes
$$
\sqrt{d}(1- \ga)=\sqrt {b}C\cos\sqrt {b}q.
$$
The both side of the last equality is positive.
Hence it is equivalent to
$$
d(1- \ga)^2=b(C^2- \ga^2).
$$
Substituting (\re{ab}) we obtain
$1- \ga=C^2/ \ga- \ga$.
Then
$$
C=\sqrt{ \ga},\quad A=( \ga-1)e^{\sqrt{ \ga/(1- \ga)}\arcsin\sqrt \ga}
$$
and
\be\la{q-sol}
q=\sqrt\ga\arcsin\sqrt\ga.
\ee
\subsection{Linearized equation }
Let us linearize equation (\re{Se}) with 
$F(\psi)=F_0(\psi)=-U_0(\psi)$ 
at the kink $s_0(x)$
splitting  the solution as the sum
\be\la{dec}
\psi(t)=s_0+\phi(t),
\ee
Substituting (\re{dec}) to (\re{Se}), we obtain
\be\la{addeq}
\ddot\phi(x,t)=\phi''(x,t)- U'_0(s_0(x)+\phi(x,t))+U'_0(s_0(x)).
\ee
By (\ref{pp}) we can write  equations (\re{addeq}) as
$$
\ddot\phi(t)=-H_0\phi(t)+{\cal N}(\phi(t)),\,\,\,t\in\R
$$
where  ${\cal N}(\phi)$ is at least quadratic in $\phi$.
and
\be\la{VV}
H_0=-\fr{d^2}{dx^2}+W_0(x),~~~~~~~~
W_0(x)=U_0''(s_0(x))=
\left\{
\ba{rl}
-b,&|x|\le q\\
\\
d,&|x| > q
\ea\right.
\ee
(see Fig. \re{Fig1}).
\begin{figure}[!ht]
\vspace{0cm}
\centering
\hspace{0cm}
\includegraphics[width=0.7\textwidth]{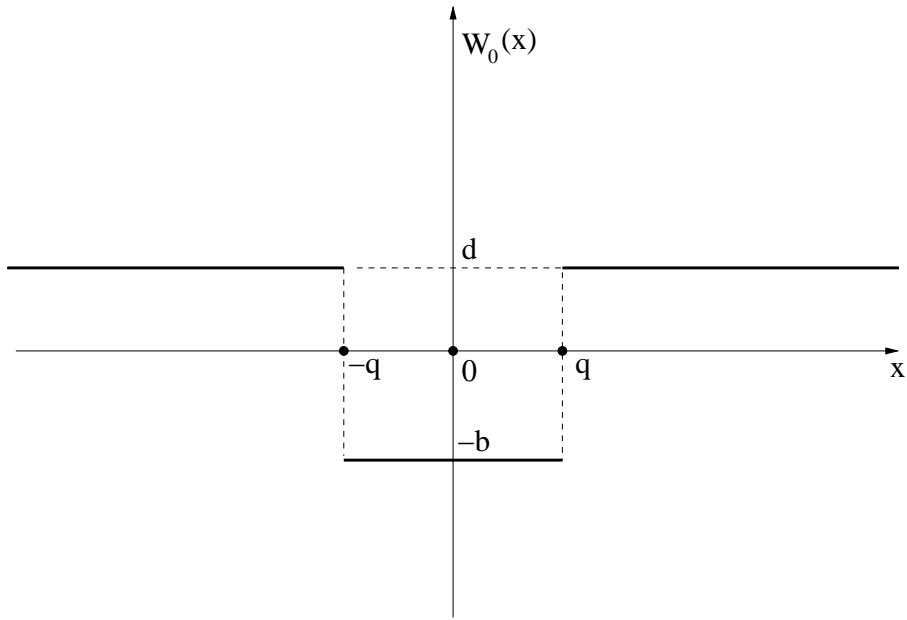}
\caption{Potential $W_0$}
\label{Fig1}
\end{figure}
\setcounter{equation}{0}
\section{Spectrum of linearized equation}
The continuous spectrum $\si_c H_0=[d,\infty)$. 
The point  $\lam_0=0$ is the groundstate
since it  corresponds to the symmetric positive
eigenfunction  $\vp_0(x)=s_0'(x)$:
$$
H_0\vp_0=-s_0'''(x)+U_0''(s_0(x))s_0'(x)=0,
$$
which follows by differentiation of (\re{ssol0}).
Therefore,  the discrete
spectrum $\si_dH_0\subset [0,d]$,
and the next eigenfunction $\vp_1(x)$ should be antisymmetric.
\subsection{Antisymmetric eigenfunctions}
The eigenfunction $\vp(x)$ corresponding to eigenvalue $\lam$
should satisfy the equation
\be\la{dei}
\left\{\ba{c}
-\vp''(x)-b\vp(x)=\lam\vp(x),~~~|x| \le q,
\\
\\
-\vp''(x)+d\vp(x)=\lam\vp(x),~~~|x| > q.
\ea\right.
\ee
Equations  (\re{dei}) imply that the antisymmetric eigenfunctions
have the form
\be\la{deif}
\vp(x)=\left\{
\ba{ll}
B\sin\beta x&\!\!\!\!\!,~~~|x| \le  q,
\\
\\
A\sgn x ~e^{-\al |x|}&\!\!\!\!\!,~~~|x| > q.
\ea\right.
\ee
where
$\al>0$, $\beta\ge 0$, and
$\al^2=d-\lam$,  $\beta^2=b+\lam$.
Let us calculate the corresponding eigenvalues $\lam$. First,
equating
the values of the eigenfunction and its first derivatives at $x=q$,
we obtain
\be\la{ABeq}
Ae^{-\al q}=B\sin\beta q,~~~~~~-A\al e^{-\al q}=B\beta\cos\beta q.
\ee
The system admits nonzero solutions if and only if its determinant vanishes:
\be\la{alb}
-\al=\beta\cot\beta q.
\ee
At last, multiplying by $q$, and denoting $\xi=\beta q$ and $\eta=\al
q$, we obtain the system
\be\la{xiet}
-\eta=\xi\cot\xi,~~~~~~\xi^2+\eta^2=R^2,
\ee
where $R=q\sqrt{b+d}$ is the radius of the circle.
Substituting $b,d$ and $q$ from  (\re{ab}) and  (\re{q-sol})
respectively, we obtain
\be\la{R2}
R=q\sqrt{\fr 1 \ga+\fr 1{1- \ga}}=\fr{q}{\sqrt{ \ga(1- \ga)}}
=\fr{\arcsin\sqrt{ \ga}}{\sqrt{1- \ga}}.
\ee
Finally, the solutions to  (\re{xiet}) can be found grafically
(see Fig. 1). 
Taking into account that $\eta> 0$, we obtain that
\be\la{gras}
\left.
\ba{rrl}
  &R\in \ds(0,~\fr\pi 2]~~&\!\!\!:
~ \mbox{no solution to (\re{xiet}) }
\\
\\
  &R\in (\ds\fr\pi 2, ~\fr{3\pi}2]&\!\!\!:
~ \mbox{exactly one solution to (\re{xiet})}
 \\
\\
  &R\in (\ds\fr{3\pi} 2, \fr{5\pi}2]&\!\!\!:
~ \mbox{exactly two solution to (\re{xiet})}
\\
\\
&& ............................................................
\ea\right|
\ee
\begin{figure}[!ht]
\vspace{0cm}
\centering
\hspace{0cm}
\includegraphics[width=0.7\textwidth]{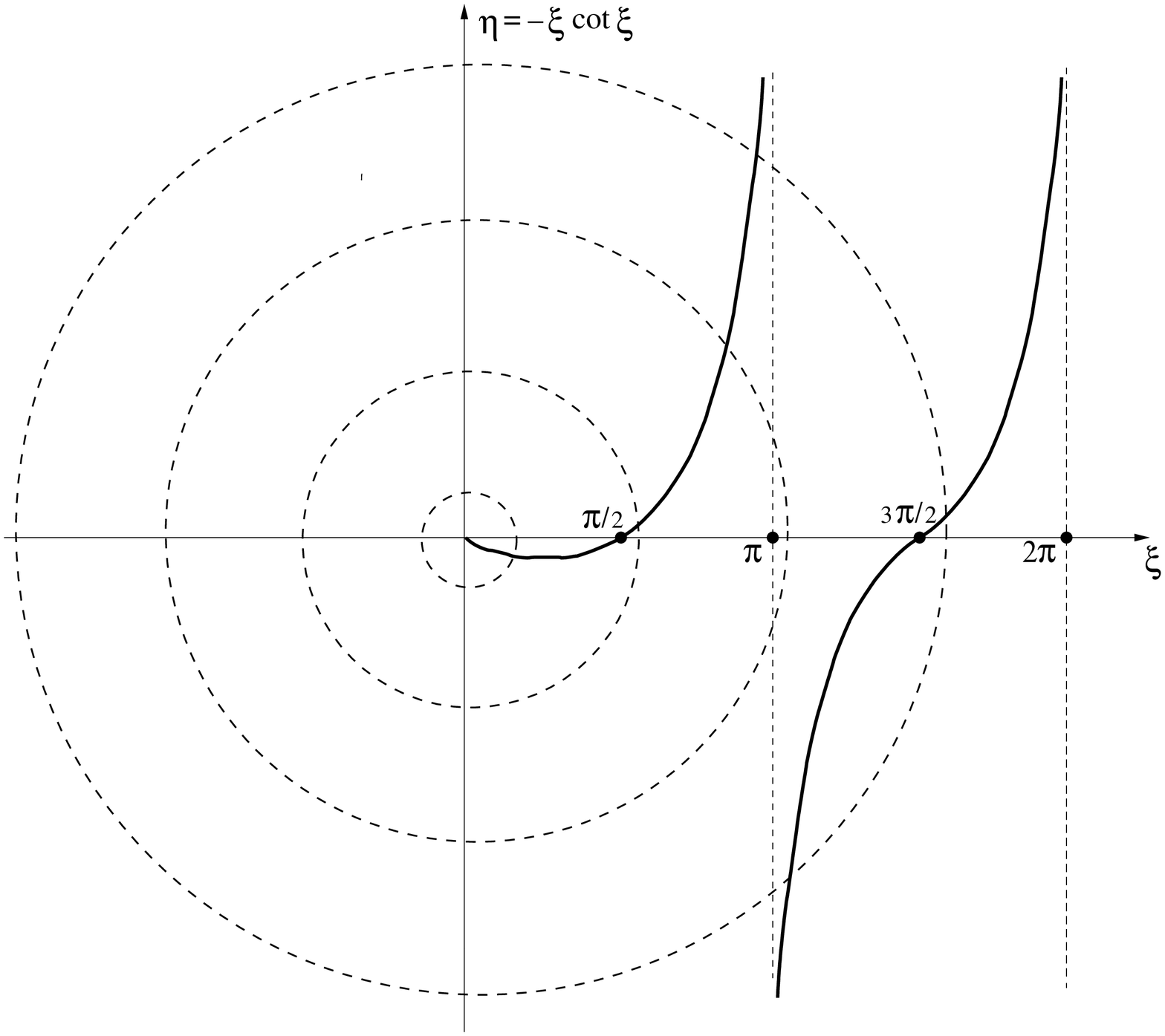}
\caption{}
\label{Picture1}
\end{figure}
Let us note that  $R(0)=0$ and $R(1)=\infty$, and
the radius $R( \ga)$ is monotone increasing on
$[0,1]$.
Denote by $ \ga_k$, $k\in\N$ the solution to the equation
\be\la{gak}
\fr{\arcsin\sqrt{ \ga_k}}{\sqrt{1- \ga_k}}=\fr{k\pi}2,\quad k\in\N.
\ee
Numerical calculations gives
\be\la{gak-sol}
\ga_1\sim 0.64643,~~\ga_2\sim 0.8579,~~\ga_3\sim 0.92472,
~~\ga_4\sim 0.95359,~~\ga_5\sim 0.96856...\;.
\ee
We have 
$ \ga_k\sim 1-\fr4{(k\pi)^2}$,  for large $k$.
Further,  (\re{gras})
implies that
\be\la{ga-int}
\left.
\ba{rrl}
&\ga\in (0,~ \ga_1]~~&\!\!\!:
~ \mbox{no nonzero antisymmetric eigenfunctions}
\\
\\
&\ga\in ( \ga_1,~ \ga_3]&\!\!\!:
~ \mbox{exactly
one linearly independent antisymmetric eigenfunctions}
\\
\\
&\ga\in ( \ga_3,~ \ga_5]&\!\!\!:
~ \mbox{exactly
two linearly independent antisymmetric eigenfunctions}
\\
\\
&& ............................................................
\ea\right|
\ee
In particular, for $\ga\in ( \ga_1, \ga_3]$  we obtain the
eigenvalue $\lam_1\in (0,d)$ corresponding to 
the antisymmetric eigenfunction:
\be\la{lam1}
\lam_1=\lam_1(\ga)=\beta^2-b=\frac{\xi^2}{q^2}-b=\fr 1{ \ga}
\Big(\fr{\xi^2}{\arcsin^2\sqrt \ga}-1\Big)
=\fr 1{ \ga}\Big(\fr{\sin^2\xi}{1- \ga}-1\Big),
\ee
where $\xi$ is the solution to
\be\la{lam2}
\fr{\xi^2}{\sin^2\xi}=\ds\fr{\arcsin^2\sqrt{ \ga}}{1- \ga}.
\ee
\subsection{Symmetric eigenfunctions}
Now we consider  symmetric eigenfunctions.
Equations  (\re{dei}) imply that the symmetric eigenfunctions
have the form
\be\la{symf}
\vp(x)=\left\{
\ba{ll}
B\cos\beta x&\!\!\!\!\!,~~~|x| \le  q,
\\
\\
A~e^{-\al |x|}&\!\!\!\!\!,~~~|x| > q,
\ea\right.
\ee
where
$\al>0$, $\beta\ge 0$, and
$\al^2=d-\lam$,  $\beta^2=b+\lam$.
Let us calculate the corresponding eigenvalues $\lam$.
Similarly (\re{ABeq})-(\re{xiet}),
denoting $\xi=\beta q$ and $\eta=\al q$, we obtain the system
\be\la{xet}
 \eta=\xi\tan\xi,~~~~~~\xi^2+\eta^2=R^2,
\ee
where $R=\ds\fr{\arcsin\sqrt{ \ga}}{\sqrt{1- \ga}}$.
\begin{figure}[!ht]
\vspace{0cm}
\centering
\hspace{0cm}
\includegraphics[width=0.7\textwidth]{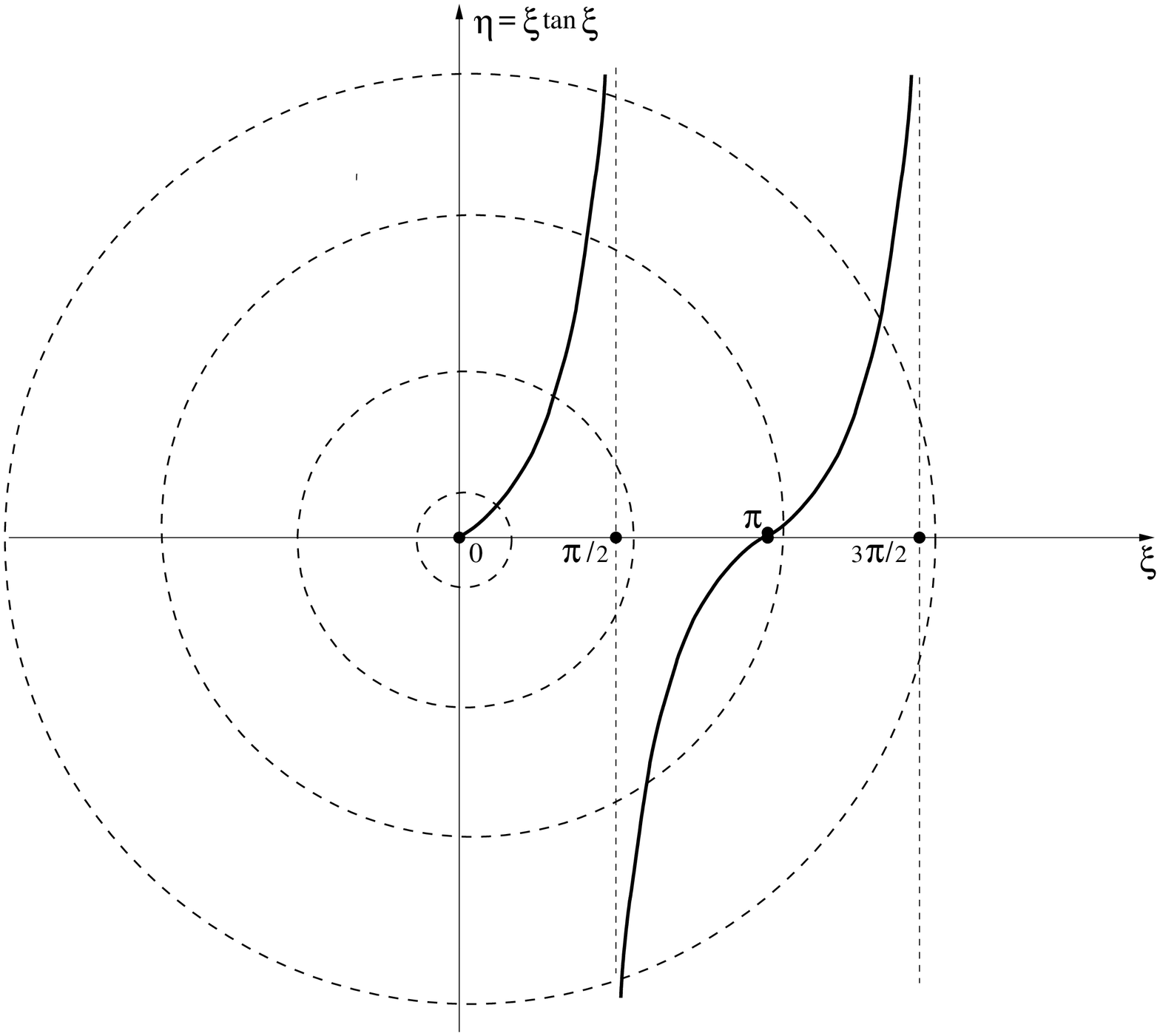}
\caption{}
\label{Picture2}
\end{figure}
The solutions to system (\re{xiet}) can be found grafically
(see Fig. 2). We have
\be\la{gras1}
\left.
\ba{rrl}
  &R\in\ds (0,~\pi]~&\!\!\!:
~ \mbox{exactly one solution to (\re{xet})}
\\
\\
& R\in (\pi, 2\pi]&\!\!\!:
~ \mbox{exactly two linearly independent symmetric eigenfunctions}
\\
\\
&& ............................................................
\ea\right|
\ee
Note that for any $\ga\in (0,1)$ equation (\re{xet}) has the solution
$\xi=\arcsin\sqrt \ga\in (0,\pi/2)$. The solution 
corresponds to eigenvalue $\lam=0$ and the first
symmetric eigenfunction.
Moreover,  (\re{gras1})
implies that
\be\la{ga1-int}
\left.
\ba{rrl}
&\ga\in (0,~ \ga_2]~~&\!\!\!:
~ \mbox{exactly
one linearly independent symmetric eigenfunctions}
\\
\\
&\ga\in ( \ga_2,~ \ga_4]&\!\!\!:
~ \mbox{exactly
two linearly independent symmetric eigenfunctions}
\\
\\
&& ............................................................
\ea\right|
\ee
where $\ga_i$ are defined in (\re{gak}).
\\
{\bf Conclusion:}\\
1) There is exactly one eigenvalue $\lam_0=0$ for $ \ga\in (0, \ga_1]$.\\
2) There are exactly two eigenvalues $\lam_0=0$ and $0<\lam_1<d$ for
$ \ga\in ( \ga_1, \ga_2]$.\\
etc.
\begin{figure}[!ht]
\vspace{0cm}
\centering
\hspace{0cm}
\includegraphics[width=0.7\textwidth]{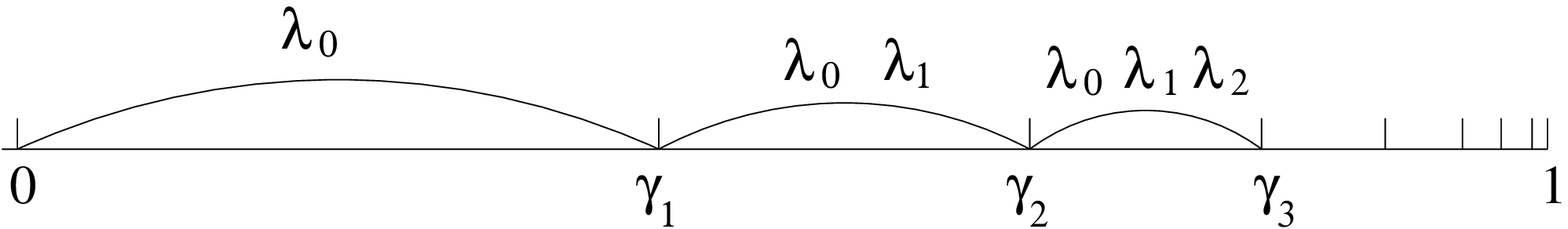}
\caption{Spectrum}
\label{Picture5}
\end{figure}
\setcounter{equation}{0}
\section{Spectral conditions} 
We deduce  Theorems \re{t1} in Section \ref{PT} below from the
following proposition.
\bp\la{p1}
For any $\ga\in(\ga_1,\ga_2)$ 
the piece wise parabolic potentials $U_0$,
defined in (\re{pp}), satisfy  conditions
{\bf U1}- {\bf U3} except for
the smoothness condition at the points $\psi=\pm\ga$.
Condition {\bf U4} holds for  any $\ga\in(\ga_1,\ga_2)$
except for one point $\ga_*$.
\ep
\begin{proof}
{\it Step i)}
Obviously, for $U_0(\psi)$ condition {\bf U1} with $a=1$, $m^2=d$ 
and any integer $K\ge 3$ holds except
the smoothness  at the points $\psi=\pm \ga$.

Consider condition {\bf U2}.
Note that the solutions to (\re{xiet}) or (\re{xet}) with $\eta=0$
and $R=k\pi/2$, $k\in\N$
correspond to $\al=0$ i.e. $\lam=d$. Then the functions
(\re{deif}) and (\re{symf}) with $A\ne 0$
is a nonzero constant for $|x|\ge \ga$.
Hence, the function
is the {\it resonance} corresponding to the edge point
$\lam=d$ of the continuous spectrum.
Thus, the resonances exist only for the discrete set of parameters
$ \ga_k\in(0,1)$ defined in (\re{gak}).
Evidently, the set has just one limit point $1$.
Hence, conditions {\bf U2}
holds if  $ \ga\in ( \ga_1, \ga_2)$.
\medskip\\
{\it Step ii)}
For  any $ \ga\in ( \ga_1, \ga_2)$ 
the operator $H_0$ defined in (\re{VV})
has exactly two  eigenvalues $\lam_0=0$ and $\lam_1\in (0,d)$.
For  condition {\bf U3} it remains to verify  (\re{U2}) 
with $m^2=d$. Namely, due to (\re{lam1})-(\re{lam2})
we  must prove that for any  $\ga\in ( \ga_1, \ga_2)$
the inequality holds 
$$
\fr 4{ \ga}\Big(\fr{\sin^2\xi(\ga)}{1- \ga}-1\Big)>\fr 1{1- \ga},
$$
where $\xi(\ga)\in (\pi/2,\pi)$ is the solution to (\re{lam2}),

After the simple transformations we obtain
\be\la{eq1}
4\cos^2\xi(\ga)<3\ga,
\ee
and then
$$
\fr {\pi}2<\xi(\ga)< \pi-\arccos\fr{\sqrt{3 \ga}}2.
$$
Since $\ds\fr{\xi}{\sin\xi}$ is monotonically increasing function
for $\xi\in (\pi/2,\pi)$, then 
$$
\fr{\pi}2<\fr{\arcsin\sqrt{ \ga}}{\sqrt{1- \ga}}
<\fr{2(\pi-\arccos\fr{\sqrt{3 \ga}}2)}{\sqrt{4-3 \ga}}.
$$
Finally, we obtain
$$
 \ga_1< \ga< \al,
$$
where $ \al$ is the solution to
$$
\fr{\arcsin\sqrt{ \al}}{\sqrt{1- \ga_{0}}}
=\fr{2(\pi-\arccos\fr{\sqrt{3 \al}}2)}{\sqrt{4-3 \al}}.
$$
Numerical calculation gives 
$$
\al=0.921485>\ga_2.
$$
Therefore, condition {\bf U3} holds for any  $\ga\in(\ga_1,\ga_2)$.
\medskip\\
{\it Step iii)}
Finally, consider condition {\bf U4} (Fermi Golden Rule).
The condition can be rewritten as
\be\la{FGR2}
  \int U_0'''(s_0(x))\vp_{4\lam_1}(x)\vp_{\lam_1}^2(x)dx
  =\int \fr{d}{dx}U_0''(s_0(x))\fr{\vp_{4\lam_1}(x)\vp_{\lam_1}^2(x)}
  {s'_0(x)}dx\ne 0.
\ee
By  (\re{VV}) we have that $U_0''(s_0(x))=W_0(x)$
is the piece wise constant function.
Hence,  
$$
\fr{d}{dx}U_0''(s_0(x))=(b+d)\de(x-q)-(b+d)\de(x+q),
$$ 
and (\re{FGR2}) becomes
$$
\vp_{4\lam_1}(q)\vp_{\lam_1}^2(q)\ne 0.
$$
Formula (\re{deif}) yields that $\vp_{\lam_1}(q)=Ae^{-\al q}\ne 0$.
Hence it is sufficient to verify that
$$
\vp_{4\lam_1}(q)\ne 0.
$$
The  eigenfunction $\vp_{4\lam_1}$ satisfies the equations
\be\la{4lam}
\left\{\ba{l}
-\vp_{4\lam_1}''(x)-b\vp_{4\lam_1}(x)=4\lam_1\vp_{4\lam_1}(x),~~~|x| \le q,\\
\\
-\vp_{4\lam_1}''(x)+d\vp_{4\lam_1}(x)=4\lam_1\vp_{4\lam_1}(x),~~~|x| > q.
\ea\right.
\ee
For the odd solution to (\re{4lam}) we have
$\vp_{4\lam_1}(x)=\sin\beta x$, $|x| \le  q$, where
$\beta^2=b+4\lam_1> 0$.
Therefore,
$$
\vp_{4\lam_1}(q)=\sin\beta q=0
$$
if either $\beta q=k\pi$, $k=0,1,2,..$, or
\be\la{sq}
\sqrt{1+4 \ga\lam_1( \ga)}\arcsin\sqrt \ga=k\pi,~~k=0,1,2,...
\ee
where $\lam_1( \ga)$ is defined in (\re{lam1})-(\re{lam2}).
Substituting $\lam_1( \ga)$ into
(\re{sq}) we obtain from (\re{lam1})-(\re{lam2})
\be\la{eqF}
\left\{\ba{l}
\ds\frac{\arcsin\sqrt{\ga}}{\sqrt{1-\ga}}
\sqrt{4\sin^2\xi-3(1-\ga)}=k\pi
\\\\
\ds\fr{\xi^2}{\sin^2\xi}=\ds\fr{\arcsin^2\sqrt{ \ga}}{1- \ga}.
\ea\right.
\ee
For $\ga\in (\ga_1,\ga_2)$
this system has a solution only for $k=1$ since
$$
0<\frac{\arcsin\sqrt{\ga}}{\sqrt{1-\ga}}
\sqrt{4\sin^2\xi-3(1-\ga)}<2\pi,\quad \ga_1<\ga<\ga_2
$$
Let us prove that (\re{eqF}) with $k=1$ has a unique solution.
Denote 
\be\la{d}
\theta=\arcsin\sqrt\ga\in (\pi\sqrt{1-\ga_1}/2,\pi\sqrt{1-\ga_2}).
\ee 
Then
(\re{eqF}) with $k=1$ is  equivalent to
\be\la{eqF1}
\left\{\ba{l}
4\xi^2-3\theta^2=\pi^2
\\\\
\ds\fr{\sin\xi}{\xi}=\ds\fr{\cos\theta}{\theta}
\ea\right.
\ee
The function 
$\theta_1(\xi):=\ds\fr{1}{\sqrt 3}\sqrt{4\xi^2-\pi^2}$
increases for 
$\xi(\ga_1)<\xi<\xi(\ga_2)$,
and
\be\la{d1}
\theta_1'(\xi)=\fr{1}{\sqrt 3}\fr{4\xi}{\sqrt{4\xi^2-\pi^2}}
>\fr{1}{\sqrt 3}\fr{4(\pi/2)}{\sqrt{4(3\pi/4)^2-\pi^2}}
=\fr{4}{\sqrt 15}>1,\quad\xi(\ga_2) <\xi<\xi(\ga_2)
\ee
since $\xi(\ga_1)=\pi/2$ and  $\xi(\ga_2)\sim 2.3137<3\pi/4$.\\
On the other hand, denote  $\theta_2:=\theta_2(\xi)$ the solution of 
$\ds\fr{\sin\xi}{\xi}=\ds\fr{\cos\theta}{\theta}$.
We have
\be\la{d2}
\theta_2'(\xi)=\fr{\sin\xi-\xi\cos\xi}{\xi^2}
\fr{\theta^2}{\cos\theta+\theta\sin\theta}>0,\quad\pi/2<\xi<\xi(\ga_2).
\ee
\begin{figure}[!ht]
\vspace{0cm}
\centering
\hspace{0cm}
\includegraphics[width=0.7\textwidth]{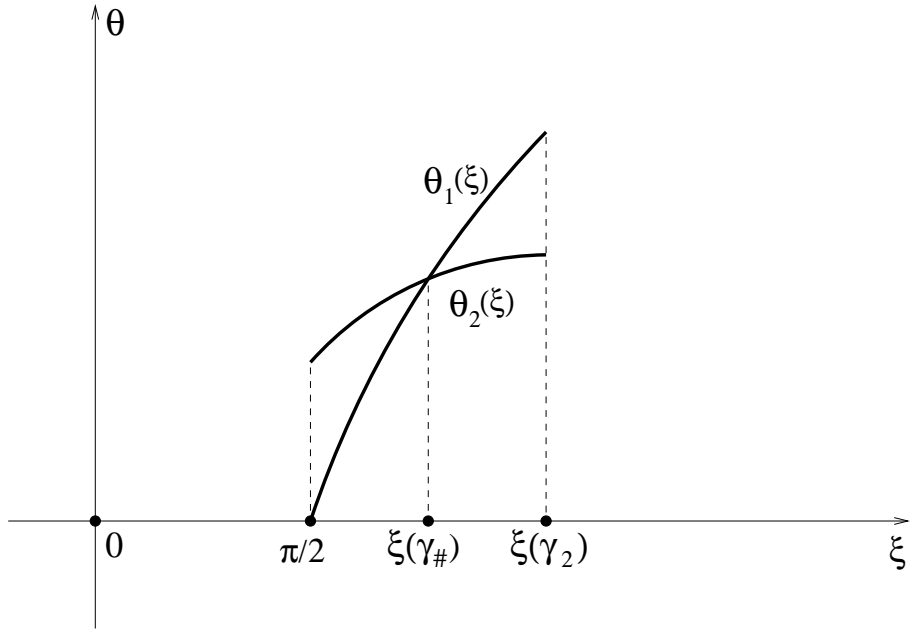}
\caption{Functions $\theta_1$ and $\theta_2$ }
\label{Fig3}
\end{figure}
Moreover, by (\re{d}) and  (\re{eqF1}) we obtain
\be\la{d3}
\theta_2'(\xi)=\fr{\theta}{\xi}\;\fr{\fr{\sin\xi}{\xi}-\cos\xi}
{\fr{\cos\theta}{\theta}+\sin\theta}
<\fr{\fr{\sin\xi}{\xi}-\cos\xi}
{\fr{\sin\xi}{\xi}+\sin\theta}<1,
\quad \pi/2<\xi<\xi(\ga_2)
\ee
since  
$|\cos\xi|<|\cos\xi(\ga_2)|<\sqrt 2/2$, and
$\sin\theta=\sqrt\ga>\sqrt\ga_1>\sqrt 2/2$ by (\re{gak-sol}).
Finally,
\be\la{d4}
\theta_2(\pi/2)>\theta_1(\pi/2)=0,
\quad
\theta_2(\xi(\ga_2))\sim 1.1843
<\theta_1(\xi(\ga_2))\sim 1.9616.
\ee
Therefore, (\re{d1})-(\re{d4}) imply that
$\theta_1(\theta)=\theta_2(\theta)$ for a single value 
$\xi(\ga_*)\in (\pi/2,\xi(\ga_2))$
(see. Figure 5).
Numerical calculation gives $\ga_*\sim 0.7925$.
Hence, system (\re{eqF}) on the interval $(\ga_1,\ga_2)$ 
has the  solution only for $\ga=\ga_*$. 
Thus, the Fermi Golden Rule holds for any $\ga\in (\ga_1, \ga_2)$
except the only point $\ga_*$.
\medskip\\
{\bf Conclusion:}\\
The potential $U_0(\psi)$ satisfies conditions {\bf U1}-{\bf U4}
except
the smoothness condition at the points $\psi=\pm \ga$
for any 
$\ga\in (\ga_1,\ga_*)\cup(\ga_*,\ga_2)$.
\end{proof}

\setcounter{equation}{0}
\section{Smooth potentials}\la{PT}
We deduce Theorem \re{t1}  from Proposition \re{p1}
by an approximation of the potential (\re{VV}) with a smooth functions
satisfying conditions {\bf U1}-{\bf U4}.
Namely, let $h(\psi)\in C_0^\infty(\R)$
be an even  mollifying function with
the following properties:
\be\la{h}
h(\psi)\ge 0,\quad\supp h\subset[-1,1],\quad
\int h(\psi)d\psi=1.
\ee
For $\ve\in(0,1]$ we define the approximations 
\be\la{Ueps}
\ti U_\ve(\psi):=\fr1\ve\int h(\ds\fr{\psi-\psi'}{\ve})U_0(\psi')d\psi'.
\ee
 This is a smooth  even function, and
it is positive and symmetric w.r.t. points $\psi=\pm 1$
in a small neighborhood of these points
for $\ve<\ga$. 
More precisely, the difference 
$$
\ti U_\ve(\psi)-U_0(\psi)=
\left\{\ba{rl}
\mu_\ve>0,& |\psi|\ge \ga+\ve,\\
\\
-\nu_\ve<0,& |\psi|\le \ga-\ve,
\ea\right.
$$
where $\mu_\ve,\nu_\ve=\cO(\ve^2)$.
Let us set 
\be\la{Ueps1}
U_\ve(\psi)=\ti U_\ve(\psi)-\mu_\ve.
\ee 
Then
\be\la{dif2}
U_\ve(\psi)=
\left\{\ba{ll}
U_0(\psi),& |\psi|\ge \ga+\ve,\\
\\
U_0(\psi)-\mu_\ve-\nu_\ve,& |\psi|\le \ga-\ve.
\ea\right.
\ee
Obviously,
\be\la{UU2}
\sup_{\psi\in\R}|U_\ve(\psi)-U_0(\psi)|\le C\ve
\ee
with some constant $C$. Moreover,
\be\la{Ueps3}
U_\ve'''(\psi)\le0~~~~\mbox{for}~~\psi\le 0,~~~~~~~~
U_\ve'''(\psi)\ge0~~~~\mbox{for}~~\psi\ge 0.
\ee
The corresponding kinks  are the odd  solutions to the equation
$$
s_\ve''(x)-U_\ve'(s_\ve(x))=0,~~~~~~x\in\R.
$$
The equation can be integrated using the ``energy conservation''
$$
\fr{|s_\ve'(x)|^2}2-U_\ve(s_\ve(x))=\const,~~~~~~x\in\R
$$
with $\const=0$:
\be\la{ssolint}
\int_0^{s_\ve(x)}
\fr{ds}{\sqrt{2U_\ve(s)}}
=x,~~~~~~x\in\R.
\ee
Hence, $s_\ve(x)$ is a monotone increasing function, and
$$
s_\ve(x)\to \pm 1,~~~~~~~x\to\pm\infty.
$$
Moreover, (\re{dif2}), (\re{UU2}) and (\re{ssolint}) imply that
$$
\sup_{x\in\R}|s_\ve(x)-s_0(x)|\le C_1\ve.
$$
Therefore,
$$
||s_\ve(x)|-\ga|\ge \ve~~~~~~~~\mbox{for}~~~~~||x|-q|\ge \de
$$
where 
\be\la{del}
\delta\to 0~~~{\rm as}~~~\ve\to 0.
\ee
Hence,
$$
W_\ve(x):=U_\ve''(s_\eps(x))
=W_0(x)~~~~~~~~
\mbox{for}~~~~~||x|-q|\ge\de
$$
and
$$
W_\ve'(x)\le 0 ~~~~\mbox{for}~~x\le 0,~~~~~~~~
W_\ve'(x)\ge 0 ~~~~\mbox{for}~~x\ge 0
$$
by (\re{Ueps3}) (see Fig. \re{Fig2}).

\begin{figure}[!ht]
\vspace{0cm}
\centering
\hspace{0cm}
\includegraphics[width=0.8\textwidth]{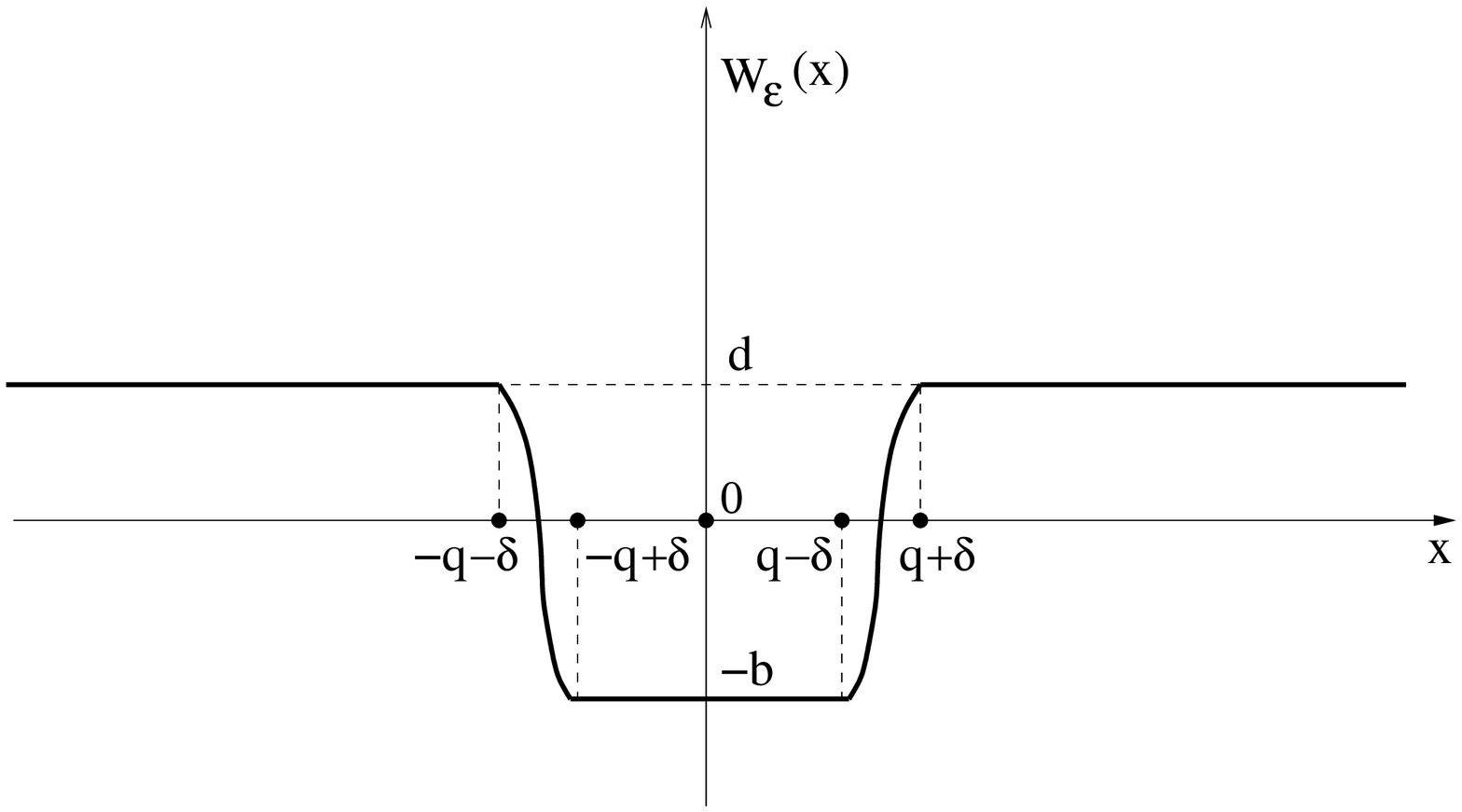}
\caption{Potential $W_\ve$}
\label{Fig2}
\end{figure}
As a result,
\be\la{w4}
W_\ve(x)-W_0(x)=0 ~~~~\mbox{for}~~||x|-q|\ge\de,
~~~~~~|W_\ve(x)-W_0(x)|\le b+d~~~~\mbox{for}~x\in\R.
\ee
Hence, denoting
$w_\ve(x)=W_\ve(x)-W_0(x)$, we obtain
\be\la{R5}
\Vert w_\ve\Vert_{L^2(\R)}\to 0,
~~~~~~\ve\to 0
\ee
by (\re{w4}) and (\re{del}).
\begin{lemma}\la{ev}
The eigenvalues
of the Schr\"odinger operator
\be\la{Seps}
H_\ve=-\fr{d^2}{dx^2}+W_\ve(x)
\ee
converge to the ones of $H_0$ as $\ve\to 0$.
\end{lemma}
\begin{proof}
The eigenvalues of  $H_0$ and $H_\ve$
are the poles of the  resolvents $R_0(\om)=(H_0-\om)^{-1}$
and $R_\ve(\om)=(H_\ve-\om)^{-1}$  respectively.
Hence, the lemma follows from (\re{R5}) due to the relation 
\be\la{R1}
R_\ve(\om)=(H_0-\om+w_\ve)^{-1}=R_0(\om)(1+w_\ve R_0(\om))^{-1}.
\ee
\end{proof}
\hspace{-0.6 cm}{\bf Proof of Theorem \re{t1}}
Consider the potential $U(\psi)=U_\ve(\psi)$ 
defined in (\re{Ueps})-(\re{Ueps1}).
Let us prove that there exist $\ve_0>0$ such that 
for any $\ga\in (\ga_1, \ga_2)$, and $0<\ve<\ve_0$
the potential $U_\ve$ satisfies conditions {\bf U1}- {\bf U4}.
\medskip\\
{\it Step i)}
Condition {\bf U1} with $a=1$, $m^2=d$ and any integer $K\ge 3$ 
obviously holds.
\medskip\\
{\it Step ii)}
For $\si\in\R$, and $s=0,1,2,...$ denote by $\cH^s_\si=\cH^s_\si (\R)$
the weighted Sobolev spaces with the finite norms
$$
\Vert\psi\Vert_{\cH^s_\si}=\sum\limits_{k=0}^{s}
\Vert (1+|x|)^{\si}\psi^{(k)}\Vert_{L^2(\R)}<\infty,
$$ 
By \ci[Theorem 7.2]{M},
the absence of the resonance at the point $\om=d$
for the Schr\"odinger operator $H_\ve$
is equivalent  to the boundedness
of the corresponding resolvent $R_\ve(\om):\cH^{0}_\si\to \cH^2_{-\si}$ 
at  $\om=d$ for any $\si>1/2$.
Hence, the resolvent
$R_0(d):\cH^{0}_\si\to \cH^2_{-\si}$ is bounded by Proposition \re{p1}.
Further, (\re{w4}) imply
$$
\Vert w_\ve\Vert_{\cH^0_{-\si}\to \cH^0_{\si}}\to 0,\quad\ve\to 0
$$
Hence, for sufficiently small $\ve$ the operator
$R_\ve(d):\cH^{0}_\si\to \cH^2_{-\si}$ is bounded  by (\re{R1}).
Then condition  {\bf U2} holds for  $U_\ve$.  
\medskip\\
{\it Step iii)} 
Lemma \re{ev} implies that for $\ga\in (\ga_1, \ga_2)$  
and sufficiently small $\ve$  
the operator $H_\ve$ has exactly two eigenvalues 
$\lam_0=0$ and $0<\lam_1(\ve)<d$. Moreover, 
$\lam_1(\ve)\to\lam_1(0)=\lam_1$ as $\ve\to 0$ and then
$4\lam_1(\ve)>d$ for sufficiently small $\ve$. 
Hence, condition {\bf U3} holds.
\medskip\\
{\it Step iv)}
It remains to check  condition {\bf U4}. Denote
$\vp^\ve_{\lam_1(\ve)}$ and $\vp^{\ve}_{4\lam_1(\ve)}$
the corresponding odd eigenfunctions of $H_\ve$. 
Then
we have 
\beqn\nonumber
 && \int U_\ve'''(s_\ve(x))\vp^\ve_{4\lam_1(\ve)}(x)
 (\vp^\ve_{\lam_1(\ve)}(x))^2dx
 =\int\limits_{|x-q|\le\de} \fr{d}{dx}W_\ve(x)
  \fr{\vp^{\ve}_{4\lam_1(\ve)}(x)(\vp^\ve_{\lam_1(\ve)}(x))^2}
  {s'_\ve(x)}dx\\
 \nonumber\\
\nonumber
 &&=\sum\limits_{\pm}\Big(d
 \fr{\vp^\ve_{4\lam_1(\ve)}(\pm q+\de)
 (\vp^\ve_{\lam_1(\ve)}(\pm q+\de))^2}
  {s'_\ve(\pm q+\de)}+b
  \fr{\vp^\ve_{4\lam_1(\ve)}(\pm q-\de)
  (\vp^\ve_{\lam_1(\ve)}(\pm q-\de))^2}
  {s'_\ve(\pm q-\de)}\Big)\\
\nonumber\\
\nonumber
&&-\int\limits_{|x-q|\le\de}W_\ve(x) \fr{d}{dx}
  \fr{\vp_{4\lam^\ve_1(\ve)}(x)
  (\vp^\ve_{\lam_1(\ve)}(x))^2}{s'_\ve(x)}dx\\
\nonumber\\
\nonumber
&&  \ba{c}
~~~\\
\longrightarrow\\\ve\to 0
\ea 2(d+b)\fr{\vp_{4\lam_1}(q)\vp_{\lam_1}^2(q)}{s'_0(q)}
=\int U_0'''(s_0(x))\vp_{4\lam_1}(x)\vp_{\lam_1}^2(x)dx
\not = 0
\eeqn
since $\de\to 0$ as $\ve\to 0$. 
Hence, condition {\bf U4} holds for sufficiently small $\ve$.
$\hfill\Box$


\begin{thebibliography}{}


\bibitem{BP2}
V. Buslaev, G. Perelman, 
On the stability of solitary waves for nonlinear
Schr\"odinger equations,
{\em Amer. Math. Soc. Trans.} (2) {\bf 164} (1995), 75-98.

\bibitem{BS2003}
V. Buslaev, C. Sulem, 
On asymptotic stability of solitary waves for nonlinear
Schr\"odinger equations,
{\em Ann. Inst. Henri Poincar\'e, Anal. Non Lin\'eaire}
{\bf 20} (2003), no. 3, 419-475.

\bibitem{IKV05}
V. Imaikin, A. Komech, B. Vainberg, 
On scattering of solitons for the Klein-Gordon
equation coupled to a particle,
{\em Comm. Math. Phys.} {\bf 268} (2006), no. 2, 321-367.


\bibitem{KK1}
E.  Kopylova, A. Komech,
On asymptotic stability of moving kink for
relativistic Ginzburg-Landau equation,
{\em Comm. Math. Phys.} {\bf 302} (2011), no. 1, 225-252. 

\bibitem{KK2}
E. Kopylova, A. Komech,
On asymptotic stability of kink for
relativistic Ginzburg-Landau equation,
{\em Arch. Rat. Mech. Anal.} {\bf  202} (2011),
no. 2, 213-245. 



\bibitem{MW96}
J. Miller, M. Weinstein,
Asymptotic stability of solitary waves 
for the regularized long-wave equation,
{\em Comm. Pure Appl. Math.} {\bf  49} (1996),  no. 4, 399-441.

\bibitem{M}
M. Murata, 
Asymptotic expansions in time for solutions of
Schr\"odinger-type equations,
{\em J.~Funct.~Anal.} {\bf 49} (1982), 10-56.

\bibitem{PW94}
R. Pego, M.I. Weinstein, 
Asymptotic stability of solitary waves,
{\em Comm. Math. Phys.} {\bf 164} (1994), 305-349.

\bibitem{Sigal93}
I. Sigal,
Nonlinear wave and Schrödinger equations.
I: Instability of periodic and quasiperiodic solutions,
{\em Comm. Math. Phys.} {\bf  153} (1993), no. 2, 297-320.


\bibitem{SW1}
A. Soffer, M.I. Weinstein,
Multichannel nonlinear scattering for nonintegrable equations,
{\em Comm. Math. Phys.} {\bf  133} (1990), 119-146.

\bibitem{SW2}
A.~Soffer, M.I.~Weinstein, Multichannel nonlinear scattering for
nonintegrable equations. II. The case of anisotropic potentials and data,
{\em J. Diff. Equations} {\bf 98} (1992), no. 2, 376-390.

\bibitem{SW99}
A.~ Soffer, M.I.~ Weinstein, Resonances, radiation damping and instability
in Hamiltonian nonlinear wave equations,
{\em Invent. Math.}  {\bf 136} (1999), 9-74.

\bibitem{SW04}
A.~ Soffer, M.I.~ Weinstein, 
Selection of the ground states for NLS equations,
{\em Rev. Math. Phys.}  {\bf 16} (2004), no. 8, 977-1071.

\bibitem{TY02}
T.-P.  Tsai, H.-T. Yau, 
Asymptotic dynamics of nonlinear Schr\"odinger equations:
resonance-dominated and dispersion-dominated solutions,
{\em Comm. Pure Appl. Math.} {\bf  55} (2002), no. 2, 153-216.

\bibitem{Tsai2003}
T.-P.  Tsai,
Asymptotic dynamics of nonlinear Schr\"odinger
equations with many bound states,
{\em J. Diff. Equations} {\bf  192} (2003), no. 1, 225-282.

\bibitem{W85}
M.I. Weinstein,
Modulational stability of ground states of nonlinear Schr\"odinger equations,
{\em SIAM J. Math. Anal.} {\bf  16}  (1985),  no. 3, 472-491.


\end{thebibliography}
\end{document}